\documentclass[11pt]{article}
\usepackage[T1]{fontenc}
\usepackage{lmodern}
\usepackage[margin=1in]{geometry}
\usepackage[utf8]{inputenc}
\usepackage[absolute]{textpos}
\usepackage[T1]{fontenc}
\usepackage{microtype}
\usepackage[all=normal,bibliography=tight]{savetrees}
\usepackage{listings}
  \usepackage{mathrsfs}

  \usepackage{amstext,amsfonts,amsthm,amsmath,amssymb}

  \usepackage{mathtools}
  
\usepackage{graphicx,tikz}
\usepackage{comment}
\usepackage{url}
\usepackage{xspace}
\usepackage{todonotes}
\usepackage{enumerate}

\def\cqedsymbol{\ifmmode$\lrcorner$\else{\unskip\nobreak\hfil
\penalty50\hskip1em\null\nobreak\hfil$\lrcorner$
\parfillskip=0pt\finalhyphendemerits=0\endgraf}\fi} 

\newcommand{\cqed}{\renewcommand{\qed}{\cqedsymbol}}

\newtheorem{lemma}{Lemma}
\newtheorem{theorem}[lemma]{Theorem}
\newtheorem{proposition}[lemma]{Proposition}

\newtheorem{claim}{Claim}
\theoremstyle{definition}

\newcommand{\pmca}{(PMC1)}
\newcommand{\pmcb}{(PMC2)}

\newcommand{\cc}{\mathsf{cc}}

\newcommand{\Om}{\Omega}
\newcommand{\pmc}{\Omega}

\newcommand{\Mod}{\mathsf{Mod}}
\newcommand{\Quo}{\mathsf{Quo}}

\newcommand{\MM}{\mathcal{M}}

\title{Covering minimal separators and potential maximal cliques in $P_t$-free graphs\thanks{
The early stage of this research was done while Andrzej Grzesik held a post-doc position of Warsaw Centre of Mathematics and Computer Science (WCMCS),
WCMCS supported a visit of Tereza Klimo\v{s}ov\'a in Warsaw,
Tereza Klimo\v sov\'a was supported by ANR project Stint under reference ANR-13-BS02-0007 and by the LABEX MILYON (ANR-10-LABX-0070) of Universit\'e de Lyon, within the program “Investissements
d’Avenir” (ANR-11-IDEX-0007) operated by the French National Research Agency (ANR).
Ma. Pilipczuk was supported by the Polish National Science Centre grant UMO-2013/09/B/ST6/03136, and
Mi.\ Pilipczuk was supported by the Foundation for Polish Science (FNP) via the START stipend programme.
Later stages of this research are parts of projects that have received funding from the European Research Council (ERC) under the European Union's Horizon 2020 research and innovation programme
Grant Agreements no.~648509 (A. Grzesik), no.~714704 (Ma. Pilipczuk), and no.~677651 (Mi.~Pilipczuk),
while 
Tereza Klimo\v sov\'a is supported by the grant no.~19-04113Y of the Czech
Science Foundation (GA\v{C}R) and the Center for Foundations of Modern Computer Science (Charles Univ. project UNCE/SCI/004).
}}

\author{ 
  Andrzej Grzesik\thanks{
    Faculty of Mathematics and Computer Science, Jagiellonian University, Krak\'ow, Poland, \texttt{andrzej.grzesik@uj.edu.pl}.
  }
  \and 
  Tereza Klimo\v{s}ov\'a\thanks{
    Department of Applied Mathematics, Charles University,  Prague, Czech Republic, \texttt{tereza@kam.mff.cuni.cz}.
  }
  \and 
  Marcin Pilipczuk\thanks{
    Institute of Informatics, University of Warsaw, Poland, \texttt{marcin.pilipczuk@mimuw.edu.pl}.
  }
  \and 
  Micha\l{} Pilipczuk\thanks{
    Institute of Informatics, University of Warsaw, Poland, \texttt{michal.pilipczuk@mimuw.edu.pl}.
  }
}

\date{}

\begin{document}
%\pagenumbering{gobble}
%\thispagestyle{empty}

\maketitle

\begin{textblock}{20}(0, 11.0)
\includegraphics[width=40px]{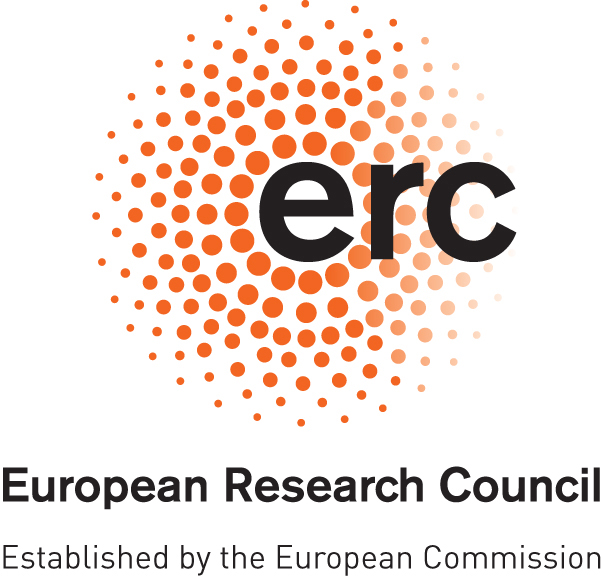}%
\end{textblock}
\begin{textblock}{20}(-0.25, 11.4)
\includegraphics[width=60px]{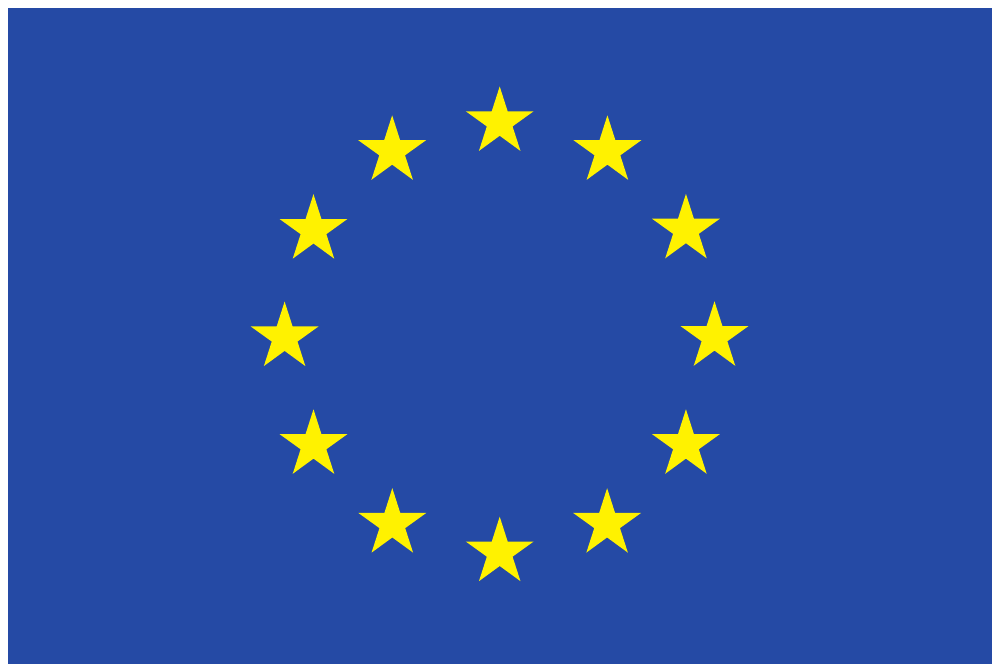}%
\end{textblock}

\begin{abstract}
A graph is called \emph{$P_t$-free} if it does not contain a $t$-vertex path as an induced subgraph. 
While $P_4$-free graphs are exactly \emph{cographs}, the structure of $P_t$-free graphs for $t \geq 5$ remains little undestood. 
On one hand, classic computational problems such as \textsc{Maximum Weight Independent Set} (\textsc{MWIS}) and \textsc{$3$-Coloring} are not known to be NP-hard
on $P_t$-free graphs for any fixed $t$.
On the other hand, despite significant effort, polynomial-time algorithms for \textsc{MWIS} in $P_6$-free graphs~[SODA 2019] and \textsc{$3$-Coloring} in $P_7$-free graphs~[Combinatorica 2018] have been found  only recently. 
In both cases, the algorithms rely on deep structural insights into the considered graph classes.

One of the main tools in the algorithms for \textsc{MWIS} in $P_5$-free graphs~[SODA 2014] and in $P_6$-free graphs~[SODA 2019] is the so-called \emph{Separator Covering Lemma} that asserts
that every minimal separator in the graph can be covered by the union of neighborhoods of a constant number of vertices. 
In this note we show that such a statement generalizes to $P_7$-free graphs and is false in $P_8$-free graphs. 
We also discuss analogues of such a statement for covering potential maximal cliques with unions of neighborhoods.
\end{abstract}

%\clearpage

%\pagenumbering{arabic}

\section{Introduction}\label{sec:intro}
By $P_t$ we denote a path on $t$ vertices. A graph is \emph{$H$-free} if it does not contain 
an induced subgraph isomorphic to $H$. 

We are interested in classifying the complexity of fundamental computational problems,
such as \textsc{Maximum Weight Independent Set} (\textsc{MWIS}), 
\textsc{$k$-Coloring} for fixed or arbitrary $k$, or \textsc{Feedback Vertex Set},
on various hereditary graph classes, in particular on $H$-free graphs for small fixed graphs $H$.
As noted by Alekseev~\cite{Alekseev82}, \textsc{MWIS} is NP-hard on $H$-free graphs
unless every connected component of $H$ is a tree with at most three leaves. 
Similarly, \textsc{$3$-Coloring} is known to be NP-hard on $H$-free graphs
unless every connected component of $H$ is a path~\cite{GolovachPS14}. 
On the other hand, it would be consistent with our knowledge if both \textsc{MWIS} and \textsc{$3$-Coloring} were polynomial-time
solvable on $P_t$-free graphs for every fixed $t$, however this is currently unknown.
Positive results in this direction are limited only to small values of $t$, as explained next.

$P_4$-free graphs, known also as {\em{cographs}}, are well-understood; in particular,
they have bounded cliquewidth, which implies the existence polynomial-time algorithms 
for all the discussed problems. 
$P_5$-free graphs are much more mysterious and only in 2014, Lokshtanov, Vatshelle, and Villanger
proposed a novel technique that uses the framework of potential maximal cliques and proved
polynomial-time tractability of \textsc{MWIS} in this class~\cite{LokshtanovVV14}. 
This result was followed by a much more technically complex positive result
for $P_6$-free graphs~\cite{GrzesikKPP19}
and a recent algorithm for \textsc{Feedback Vertex Set} in $P_5$-free graphs~\cite{ACPRS20}.
For coloring with few colors,
the state-of-the-art are polynomial-time algorithms for \textsc{$3$-Coloring}
of $P_7$-free graphs~\cite{BonomoCMSSZ18} and
\textsc{$4$-Coloring} of $P_6$-free graphs~\cite{SpirklCZ19}.

In full generality, the so-called \emph{Gy\'{a}rf\'{a}s' path argument}
gives subexponential-time algorithms for both \textsc{MWIS} and \textsc{$3$-Coloring}
in $P_t$-free graphs for any fixed $t$~\cite{BacsoLMPTL19,Brause17,GroenlandORSSS19}.
Using the Gy\'{a}rf\'{a}s' path argument and the three-in-a-tree theorem~\cite{ChudnovskyS10},
it is possible to obtain a quasi-polynomial-time approximation scheme for \textsc{MWIS}
in $H$-free graphs whenever every connected component of $H$ is a tree with at most three leaves~\cite{ChudnovskyPPT20}. Note that this covers exactly the cases where NP-hardness is not known. 
The crucial property of $P_t$-free graphs that is used in all the works mentioned above is that, due to 
the aforementioned Gy\'{a}rf\'{a}s' path argument, every $P_t$-free graph admits
a balanced separator consisting of at most $t-1$ closed neighborhoods of a vertex.

The lack of NP-hardness results on one side, and shortage of generic algorithmic tools in $P_t$-free graphs on the other side, calls for a deeper understanding of the structure of $P_t$-free graphs for larger values of $t$. 
In this note we discuss one property that appeared important in the algorithms for \textsc{MWIS}
for $P_5$-free and $P_6$-free graphs~\cite{LokshtanovVV14,LokshtanovPL18,GrzesikKPP19}, namely
the possibiliy to cover a minimal separator with a small number of vertex neighborhoods.

Let $G$ be a graph. For a set $S \subseteq V(G)$, a connected component $A$ of $G-S$
is \emph{a full component to $S$} if $N_G(A) = S$. A set $S$ is a \emph{minimal separator}
if it admits at least two full components.
A set $F \subseteq \binom{V(G)}{2} \setminus E(G)$ is a \emph{chordal completion}
if $G+F \coloneqq (V(G), E(G) \cup F)$ is chordal (i.e., does not contain an induced subgraph 
isomorphic to a cycle on at least four vertices). 
A set $\pmc \subseteq V(G)$ is a \emph{potential maximal clique (PMC)} if there exists
an (inclusion-wise) minimal chordal completion $F$ of $G$ such that $\pmc$
is a maximal clique of $G+F$. 
Potential maximal cliques and minimal separators are tightly connected: for example, 
a graph is chordal if and only if every its minimal separator is a clique, and if 
$\pmc$ is a PMC in $G$, then for every connected component $D$ of $G-\pmc$
the set $N_G(D)$ is a minimal separator with $D$ being one of the full components. 

A framework of Bouchitt\'{e} and Todinca~\cite{BouchitteT01,BouchitteT02}, extended by Fomin, Todinca, and Villanger~\cite{FominTV15}, allows solving multiple computational problems (including \textsc{MWIS} and \textsc{Feedback Vertex Set}) on graph classes where graphs have only a polynomial number of PMCs.
While $P_5$-free graphs do not have this property, the crucial insight of the work
of Lokshtanov, Villanger and Vatshelle~\cite{LokshtanovVV14} allows modifying the framework
to work for $P_5$-free graphs and, with more effort, for $P_6$-free graphs~\cite{GrzesikKPP19}.

A simple, but crucial in~\cite{LokshtanovVV14}, insight about the structure of $P_5$-free graphs
is the following lemma.
\begin{lemma}[\cite{LokshtanovVV14}]
Let $G$ be a $P_5$-free graph, let $S$ be a minimal separator in $G$, and let $A$ and $B$
be two full components of $S$.
Then for every $a \in A$ and $b \in B$ it holds that $S \subseteq N_G(a) \cup N_G(b)$.
\end{lemma}
The above statement is per se false in $P_6$-free graphs, but the following variant
is true and turned out to be pivotal in~\cite{GrzesikKPP19}:
\begin{lemma}[\cite{GrzesikKPP19}, Lemma~20 in the arXiv version]
Let $G$ be a $P_6$-free graph, let $S$ be a minimal separator in $G$, and let $A$ and $B$
be two full components of $S$.
Then there exist nonempty sets $A' \subseteq A$ and $B' \subseteq B$ such that
$|A'| \leq 3$, $|B'| \leq 3$, and $S \subseteq N_G(A) \cup N_G(B)$.
\end{lemma}
That is, every minimal separator in a $P_6$-free graph has a dominating set of size at most $6$,
     contained in the union of two full components of this separator.

In Section~\ref{sec:sep} we extend the result to $P_7$-free graphs as follows.
\begin{theorem}\label{thm:sepcov}
Let $G$ be a $P_7$-free graph and let $S$ be a minimal separator in $G$.
Then there exists a set $S'\subseteq V(G)$ of size at most $22$ such that $S \subseteq N_G[S']$.
\end{theorem}
Theorem~\ref{thm:sepcov} directly generalizes a statement proved by Lokshtanov et al.~\cite[Theorem~1.3]{LokshtanovPL18},
  which was an important ingredient of their quasi-polynomial-time algorithm for \textsc{MWIS} in $P_6$-free
  graphs.

Section~\ref{sec:ex} discusses a modified example from~\cite{LokshtanovPL18}
that witnesses that no statement analogous to Theorem~\ref{thm:sepcov} can be true in $P_8$-free graphs.
Furthermore, observe that in the statements for $P_5$-free and $P_6$-free graphs 
the dominating set for the separator is guaranteed to be contained in two full components of the separator. This is no longer the case in Theorem~\ref{thm:sepcov} for a reason: 
in Section~\ref{sec:ex} we show examples of $P_7$-free graphs where any constant-size dominating
set of a minimal separator needs to contain a vertex from the said separator.

The intuition behind the framework of PMCs, particularly visible in 
the quasi-polynomial-time algorithm for \textsc{MWIS} in $P_6$-free graphs~\cite{LokshtanovPL18},
is that potential maximal cliques can serve as balanced separators of a graph.
Here, $X \subseteq V(G)$ is a balanced separator of $G$ if every connected component of $G-X$
has at most $|V(G)|/2$ vertices. 
The quasi-polynomial-time algorithm of~\cite{LokshtanovPL18} tried to recursively
split the graph into significantly smaller pieces by branching and deleting
as large as possible pieces of such a PMC. 
Motivated by this intuition, in Section~\ref{sec:pmc} we generalize Theorem~\ref{thm:sepcov} to 
dominating potential maximal cliques:
\begin{theorem}\label{thm:pmccov}
Let $G$ be a $P_7$-free graph and let $\pmc$ be a potential maximal clique in $G$.
Then there exists a set $\pmc' \subseteq V(G)$ of size at most $68$ such that $\pmc \subseteq N_G[\pmc']$.
\end{theorem}
Since every minimal separator is a subset of some potential maximal clique in a graph,
Theorem~\ref{thm:pmccov} generalizes Theorem~\ref{thm:sepcov}. For the same reason, our examples for $P_8$-free graphs also prohibit extending Theorem~\ref{thm:pmccov}
to $P_8$-free graphs.

\section{Preliminaries}\label{sec:prelims}

For basic graph notation, we follow the arXiv version of~\cite{GrzesikKPP19}.
We outline here only nonstandard notation that is not presented in the introduction.

For a set $X \subseteq V(G)$, by $\cc(G-X)$ we denote the family of connected components of $G-X$.
A set $A$ is {\em{complete}} to a set $B$ if every vertex of $A$ is adjacent to every vertex of $B$.

\paragraph{Potential maximal cliques.}
A set $\pmc \subseteq V(G)$ is a \emph{potential maximal clique} (PMC)
  if:
\begin{description}
\item[\pmca]\label{pmca} none of the connected components of $\cc(G-\Om)$ is full to $\Om$; and
\item[\pmcb]\label{pmcb} whenever $uv$ is a non-edge with $\{u,v\}\subseteq \Om$, then there is a component $D\in \cc(G-\Om)$ such that $\{u,v\}\subseteq N(D)$.
\end{description}
In the second condition, we will say that the component $D$ {\em{covers}} the non-edge $uv$.
As announced in the introduction, we have the following.
\begin{proposition}[Theorem 3.15 of~\cite{BouchitteT01}]
For a graph $G$, a vertex subset $\Omega\subseteq V(G)$ is a PMC if and only if there exists a minimal chordal completion $F$ of $G$ such that $\Omega$ is a maximal clique in $G+F$.
\end{proposition}
We will also need the following statement.
\begin{lemma}[cf. Proposition~8 of the arXiv version of~\cite{GrzesikKPP19}]\label{lem:ND}
For every PMC $\pmc$ of $G$ and every $D \in \cc(G-\pmc)$, the
set $N_G(D)$ is a minimal separator.
\end{lemma}

\paragraph{Modules.}
Let $G$ be a graph. A set $M\subseteq V(G)$ is a {\em module} of $G$ if $N(x)\setminus M=N(y)\setminus M$ for every $x,y\in M$. 
Note that $\emptyset$, $V(G)$ and all the singletons $\{x\}$ for $x \in V(G)$ are modules; we call these modules {\em trivial}. A graph is {\em prime} if all its modules are trivial.
A module $M$ of $G$ is {\em strong} if $M \neq V(G)$ and $M$
does not overlap with any other module of $G$, i.e., for every module $M'$ of $G$ we have
either $M\subseteq M'$, or $M'\subseteq M$, or $M \cap M'=\emptyset$.

A partition $\MM$ of $V(G)$ is a {\em modular partition} of $G$ if $M$ is a module of $G$ for every $M\in \MM$. 
The {\em quotient graph} $G/\MM$ is a graph with the vertex set $\MM$ and with $M'M\in E(G/\MM)$ if and only if $m'm\in E(G)$ for all $m'\in M'$ and $m\in M$
(since $M'$ and $M$ are modules, $m'm$ is an edge either for all pairs $m'\in M'$ and $m\in M$, or for none). 

It is well-known (cf.~\cite[Lemma~2]{HabibP10}) that if $|V(G)| > 1$ 
then the family of (inclusion-wise) maximal strong modules of $G$
forms a modular partition of $G$ whose quotient graph is either an independent set
(if $G$ is not connected), a clique (if the complement of $G$ is not connected),
  or a prime graph (otherwise).
  We denote this modular partition by $\Mod(G)$ and we let $\Quo(G)\coloneqq G/\Mod(G)$.
For $D \subseteq V(G)$, we abbreviate $\Mod(G[D])$ and $\Quo(G[D])$ with $\Mod(D)$ and $\Quo(D)$, respectively.

\section{Covering minimal separators in $P_7$-free graphs}\label{sec:sep}

This section is devoted to the proof of Theorem~\ref{thm:sepcov}. 
We need the following two results from~\cite{GrzesikKPP19}.
\begin{lemma}[Bi-ranking Lemma of~\cite{GrzesikKPP19}, Lemma~17 of the arXiv version]\label{lem:biranking}
Suppose $X$ is a non-empty finite set and $(X,\leq_1)$ and $(X,\leq_2)$ are two quasi-orders.
Suppose further that every pair of two different elements of $X$ is comparable either with respect to $\leq_1$ or with respect to $\leq_2$.
Then there exists an element $x\in X$ such that for every $y\in X$ we have either $x\leq_1 y$ or $x\leq_2 y$.
\end{lemma}
\begin{lemma}[Neighborhood Decomposition Lemma of~\cite{GrzesikKPP19}, Lemma~18 of the arXiv version]\label{lem:nei-decomp}
Suppose $G$ is a graph and $D\subseteq V(G)$ is subset of vertices such that $|D|\geq 2$ and $G[D]$ is connected.
Suppose further that vertices $p,q\in D$ respectively belong to different elements $M^p,M^q$ of the modular partition $\Mod(D)$ such that $M^p$ and $M^q$ are adjacent in the 
quotient graph $\Quo(D)$.
Then, for each vertex $u\in N(D)$ at least one of the following conditions holds:
\begin{enumerate}[(a)]
\item\label{cnd:pq} $u\in N[p,q]$;
\item\label{cnd:P4} there exists an induced $P_4$ in $G$ such that
 $u$ is one of its endpoints, while the other three vertices belong to $D$;
\item\label{cnd:(} $\Quo(D)$ is a clique and the neighborhood of $u$ in $D$ is the union of some collection of maximal strong modules in $D$.
\end{enumerate}
In particular, if $\Quo(D)$ is not a clique, then the last condition cannot hold.
\end{lemma}

Let $G$ be a $P_7$-free graph, let $S$ be a minimal separator in $G$, and let $A_1$ and $A_2$
be the vertex sets of two full components of $S$. 
If $|A_i| = 1$ for some $i\in \{1,2\}$, then we are done by setting $S' \coloneqq A_i$, so assume $|A_1|,|A_2| > 1$.
For each $i\in \{1,2\}$, fix two different maximal strong modules $M^p_i$ and $M^q_i$ 
of $G[A_i]$ that are adjacent in $\Quo(A_i)$. 
Furthermore, pick arbitrary $p_i \in M^p_i$ and $q_i \in M^q_i$.

For each $i\in \{1,2\}$, we apply Lemma~\ref{lem:nei-decomp} to $D \coloneqq A_i$ and $N(D) = S$.
We say that a vertex $x \in S$ is of type $\eqref{cnd:pq}_i$ if $x \in N(p_i) \cup N(q_i)$.
We say that a vertex $x \in S$ is of type $\eqref{cnd:P4}_i$ if $x$ is not of type $\eqref{cnd:pq}_i$
and there is an induced $P_4$ in $G$ with $u$ being one of the endpoints
and the other three vertices belonging to $A_i$.
Finally, we say that a vertex $x \in S$ is of type $\eqref{cnd:(}_i$ if $x$ is neither of type
$\eqref{cnd:pq}_i$ nor $\eqref{cnd:P4}_i$. 
Lemma~\ref{lem:nei-decomp} asserts that if there are vertices of type $\eqref{cnd:(}_i$,
then $\Quo(A_i)$ is a clique and the neighborhood in $A_i$ of every vertex of this type
is the union of a collection of maximal strong modules of $G[A_i]$.
For $\alpha,\beta \in \{a,b,c\}$, let $S_{\alpha\beta}$ be the set of vertices $x \in S$
that are of type $(\alpha)_1$ and~$(\beta)_2$.

We need the following claim.
\begin{claim}\label{cl:edge}
Let $i \in \{1,2\}$ and let $x,y \in S$ be of type $\eqref{cnd:(}_i$.
Then $A_i \cap (N(x) \setminus N(y))$ is complete to $A_i \cap (N(y) \setminus N(x))$.
\end{claim}
\begin{proof}
By Lemma~\ref{lem:nei-decomp},
$\Quo(A_i)$ is a clique and both $A_i \cap (N(x) \setminus N(y))$
and $A_i \cap (N(y) \setminus N(x))$ are the unions of some disjoint collections of maximal strong modules of $G[A_i]$.
The claim follows.
\cqed\end{proof}

Since $G$ is $P_7$-free, $S_{bb} = \emptyset$. 
Furthermore, if we set $R_a \coloneqq \{p_1,q_1,p_2,q_2\}$, then
$$S_{aa} \cup S_{ab} \cup S_{ac} \cup S_{ba} \cup S_{ca} \subseteq N(R_a).$$
In the rest of the proof, we construct sets $R_{bc}, R_{cb}$ and $R_{cc}$ such that
$S_{\alpha\beta}\subseteq N[R_{\alpha\beta}]$ for $\alpha\beta \in \{bc,cb,cc\}$.
We will conclude that $S'\coloneqq R_a\cup R_{bc}\cup R_{cb}\cup R_{cc}$
satisfies the statement of the lemma, because $|R_a|=4$
and we will ensure that $|R_{bc}|,|R_{cb}|\leq 5$ and $|R_{cc}|\leq 8$.

\medskip

We start with constructing the set $R_{bc}$. 
If $S_{bc} = \emptyset$, then we set $R_{bc} = \emptyset$.
Otherwise, let $v \in S_{bc}$ be a vertex with inclusion-wise minimal
set $A_2 \cap N(v)$. 
Furthermore, let $w \in A_2$ be an arbitrary neighbor of $v$ in $A_2$;
$w$ exists since $A_2$ is a full component of $S$.
Also, let $v, u^1, u^2, u^3$ be vertices of an induced $P_4$ with $u^1,u^2,u^3 \in A_1$; recall here that $v$ is of type $\eqref{cnd:P4}_1$. 
We set $R_{bc} \coloneqq \{u^1,u^2,u^3,v,w\}$ and claim that $S_{bc} \subseteq N[S_{bc}]$.

Assume the contrary, and let $v' \in S_{bc} \setminus N[R_{bc}]$. 
By the choice of $v$ and since $w \in A_2 \cap (N(v) \setminus N(v'))$,
there exists $w' \in A_2 \cap (N(v') \setminus N(v))$. 
By Claim~\ref{cl:edge}, $ww' \in E(G)$. 
Then, $v'-w'-w-v-u^1-u^2-u^3$ is an induced $P_7$ in $G$, a contradiction.

Hence, we constructed $R_{bc} \subseteq V(G)$ of size at most $5$
such that $S_{bc} \subseteq N[R_{bc}]$.
A symmetric reasoning yields $R_{cb} \subseteq V(G)$ of size at most $5$
such that $S_{cb} \subseteq N[R_{cb}]$. 

\medskip

We are left with constructing $R_{cc}$. If $S_{cc} = \emptyset$, then we take $R_{cc} = \emptyset$
and conclude. In the remaining case, $S_{cc}$ is non-empty, so both $\Quo(A_1)$ and $\Quo(A_2)$ are cliques.

For each $i\in \{1,2\}$,
we define a quasi-order $\leq_i$ on $S_{cc}$ as follows. 
For $x,y \in S_{cc}$, $x \leq_i y$ if $A_i \cap N(x) \subseteq A_i \cap N(y)$.
An unordered pair $xy \in \binom{S_{cc}}{2}$ is a \emph{butterfly}
if $x$ and $y$ are incomparable both in $\leq_1$ and in $\leq_2$, that is, 
  if each of the following four sets is nonempty:
\begin{align*}
&A_1 \cap (N(x) \setminus N(y)), & &A_1 \cap (N(y) \setminus N(x)),\\
&A_2 \cap (N(x) \setminus N(y)), & &A_2 \cap (N(y) \setminus N(x)).
\end{align*}
\begin{figure}[tb]
\begin{center}
\includegraphics{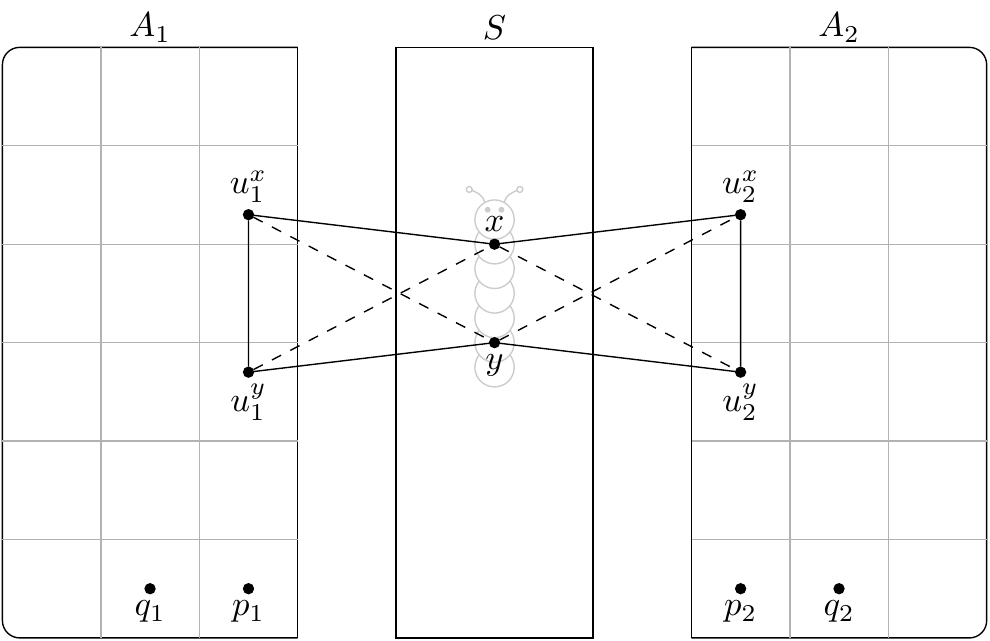}
\caption{A butterfly $xy$. The gray grid presents
  the partion of $G[A_i]$ into maximal strong modules.}\label{fig:butterfly}
\end{center}
\end{figure}
See Figure~\ref{fig:butterfly} for an illustration.

Lemma~\ref{lem:biranking} allows us to easily dominate subsets of $S_{cc}$ that do
not contain any butterflies:
\begin{claim}\label{cl:bez-motylka}
Let $T \subseteq S_{cc}$ be such that there is no butterfly $xy$ with $x,y \in T$.
Then there exist a vertex $a_1 \in A_1$ and $a_2 \in A_2$ such that $T \subseteq N(a_1) \cup N(a_2)$.
\end{claim}
\begin{proof}
If $T=\emptyset$, the claim is trivial, so assume otherwise.
Let us focus on quasi-orders $\leq_1$ and $\leq_2$, restricted to $T$.
Since there are no butterflies in $T$, the prerequisities of Lemma~\ref{lem:biranking} 
are satisfied for $(T,\leq_1)$ and $(T,\leq_2)$. 
Hence, there exists $x \in T$ with $x \leq_1 y$ or $x \leq_2 y$ for every $y \in T$.
For $i\in \{1,2\}$, let $a_i$ be an arbitrary neighbor of $x$ in $A_i$
(it exists as $A_i$ is a full component of $S$). 
For every $y \in T$, there exists $i \in \{1,2\}$ such that $x \leq_i y$,
hence $a_iy \in E(G)$. 
We conclude that $T \subseteq N(a_1) \cup N(a_2)$, as desired.
\cqed\end{proof}
If there is no butterfly at all, then we apply Claim~\ref{cl:bez-motylka} to $T = S_{cc}$,
obtaining vertices $a_1,a_2$ and set $R_{cc} = \{a_1,a_2\}$. 
Thus, we are left with the case where at least one butterfly exists.

Let $xy$ be a butterfly with inclusion-wise minimal set
$$(A_1 \cup A_2) \cap N(\{x,y\}).$$
Furthermore, pick the following four vertices
\begin{align*}
u_1^{x} \in& A_1 \cap (N(x) \setminus N(y)),& u_1^{y} \in &A_1 \cap (N(y) \setminus N(x)),\\
u_2^{x} \in& A_2 \cap (N(x) \setminus N(y)),& u_2^{y} \in &A_2 \cap (N(y) \setminus N(x)).
\end{align*}
Claim~\ref{cl:edge} ensures that $u_1^{x}u_1^{y} \in E(G)$ and $u_2^{x}u_2^{y} \in E(G)$.

Set 
$$R' \coloneqq \{x, y, u_1^{x}, u_1^{y}, u_2^{x}, u_2^{y}\}.$$
Let $T \coloneqq S_{cc} \setminus N[R']$. We claim the following:
\begin{claim}\label{cl:motylek}
There is no butterfly $x'y'$ with $x',y' \in T$.
\end{claim}
\begin{proof}
Assume the contrary, and let $x'y'$ be a butterfly with $x',y' \in T$.
By the minimality of $xy$, as $u_1^{x} \in (A_1 \cup A_2) \cap N(\{x,y\})$
but $u_1^{x} \notin (A_1 \cup A_2) \cap N(\{x',y'\})$,
there exists $w \in (A_1 \cup A_2) \cap (N(\{x',y'\}) \setminus N(\{x,y\}))$.
By symmetry, assume that $w \in A_1$ and $x'w \in E(G)$; see Figure~\ref{fig:motylek}.

By Claim~\ref{cl:edge}, $wu_1^{x} \in E(G)$ and $wu_1^{y} \in E(G)$.
If $xy \in E(G)$, then $x'-w-u_1^{x}-x-y-u_2^{y}-p_2$ would be an induced $P_7$ in $G$.
Otherwise, 
if $xy \notin E(G)$, then $x'-w-u_1^{x}-x-u_2^{x}-u_2^{y}-y$ would be an induced $P_7$ in $G$.
As in both cases we have obtained a contradiction, this finishes the proof.
\cqed\end{proof}
\begin{figure}[tb]
\begin{center}
\includegraphics{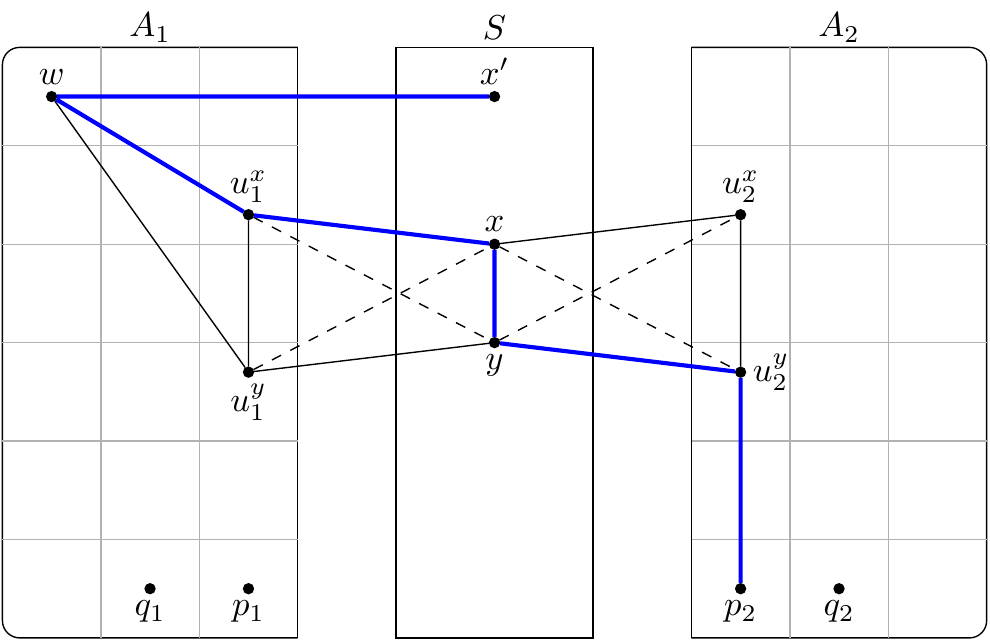}\\[3mm]%
\includegraphics{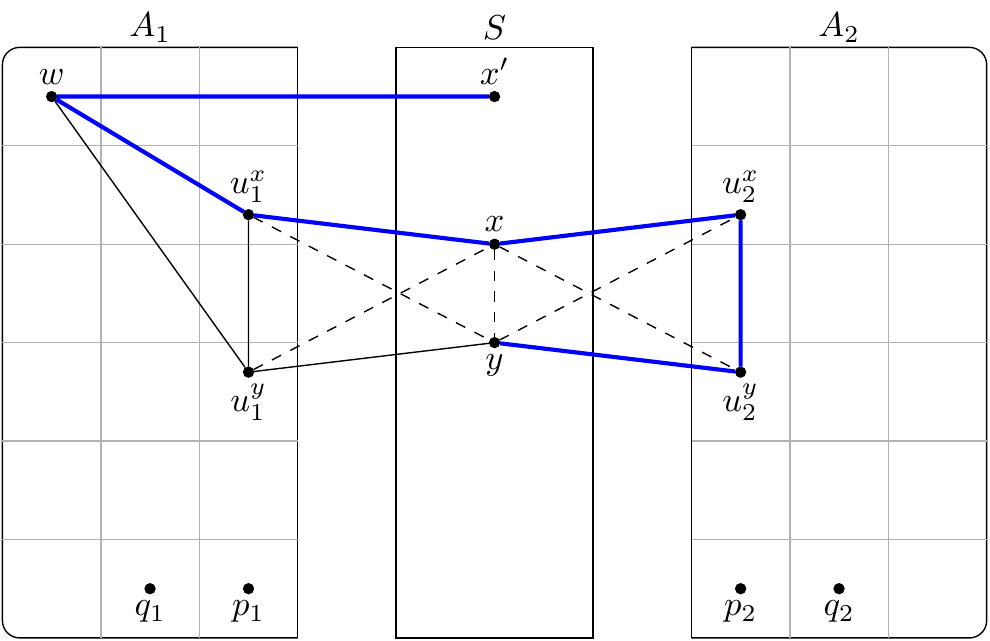}
\caption{Two cases where a $P_7$ appears in the proof of Claim~\ref{cl:motylek}.}\label{fig:motylek}
\end{center}
\end{figure}
By Claim~\ref{cl:motylek}, we can apply Claim~\ref{cl:bez-motylka}
to $T$ and obtain vertices $a_1 \in A_1$ and $a_2 \in A_2$ with $T \subseteq N(a_1) \cup N(a_2)$.
Hence, we can take $R_{cc} \coloneqq R' \cup \{a_1,a_2\}$, thus concluding the proof of Theorem~\ref{thm:sepcov}.

\section{Covering PMCs in $P_7$-free graphs}\label{sec:pmc}

We now prove the following statement which, together with Theorem~\ref{thm:sepcov}
and Lemma~\ref{lem:ND}, immediately implies Theorem~\ref{thm:pmccov}.
\begin{lemma}\label{lem:pmccov2}
Let $G$ be a $P_7$-free graph and let $\pmc$ be a potential maximal clique in $G$.
Then there exists a set $\pmc' \subseteq \pmc$ of size at most $2$
and a set $\mathcal{D}' \subseteq \cc(G-\pmc)$ of size at most $3$
such that
$$\pmc \subseteq N[\pmc'] \cup \bigcup_{D \in \mathcal{D}'} N(D).$$
\end{lemma}
\begin{proof}
Let $\mathcal{D} \subseteq \cc(G-\pmc)$ be an inclusion-wise minimal set of components
of $G-\pmc$ such that for every nonedge $uv$ in $\pmc$ there exists a component $D \in \mathcal{D}$
that covers $uv$. 

If $\mathcal{D} = \emptyset$, then $\pmc$ is a clique in $G$ and thus
we can put $\pmc' = \{v\}$ and $\mathcal{D}'= \emptyset$ for an arbitrary $v \in \pmc$.
Otherwise, pick any $D \in \mathcal{D}$. By the minimality of $\mathcal{D}$, there exists
a nonedge $uv$ in $\pmc$ that is covered by $D$ and by no other component of $\mathcal{D} \setminus \{D\}$. 

Assume that there is no component $D_v \in \mathcal{D} \setminus \{D\}$ with $v \in N(D_v)$.
Then $\pmc \setminus N(D) \subseteq N(v)$, so $\pmc \subseteq N(v) \cup N(D)$. 
Hence, we can set $\pmc' =\{v\}$ and $\mathcal{D}' = \{D\}$. 
Symmetrically, we are done if there is no component $D_u \in \mathcal{D} \setminus \{D\}$
with $u \in N(D_u)$. 

In the remaining case, pick arbitrary components 
$D_v \in \mathcal{D} \setminus \{D\}$ with $v \in N(D_v)$ and
$D_u \in \mathcal{D} \setminus \{D\}$ with $u \in N(D_u)$.
Since $D$ is the only component of $\mathcal{D}$ that covers $uv$, we have $u \notin N(D_v)$
and $v \notin N(D_u)$; in particular, $D_u \neq D_v$. 
We claim that we can set $\pmc' = \{u,v\}$ and $\mathcal{D}' = \{D,D_u,D_v\}$. That is,
   we claim that
\begin{equation}\label{eq:pmccov}
\pmc \subseteq N[u] \cup N[v] \cup N(D) \cup N(D_u) \cup N(D_v).
\end{equation}
Assume the contrary, and let $x \in \pmc$ be such that $xu \notin E(G)$, $xv \notin E(G)$,
$x \notin N(D)$, $x \notin N(D_u)$, and $x \notin N(D_v)$. 

Since $xu$ is a nonedge of $\pmc$, there exists $D_{xu} \in \mathcal{D}$ that covers $xu$.
Similarly, there exists $D_{xv} \in \mathcal{D}$ that covers $xv$.
By the choice of $x$, we have $D_{xu},D_{xv} \notin \{D,D_u,D_v\}$.
Further, since $D$ is the only component of $\mathcal{D}$ covering $uv$,
$v \notin N(D_{xu})$ and $u \notin N(D_{xv})$; in particular, $D_{xu} \neq D_{xv}$.
See Figure~\ref{fig:pmc}.

Let $y_u$ be an arbitrary neighbor of $u$ in $D_u$, let $y_v$ be an arbitrary neighbor of $v$ in $D_v$, 
let $P_u$ be a shortest path from $u$ to $x$ with all internal vertices in $D_{xu}$, and let $P_v$ be a shortest path from $v$ to $x$ with all internal vertices in $D_{xv}$. 
Then, $y_u - u - P_u - x - P_v - v - y_v$ is an induced path with at least
$7$ vertices, a contradiction. This proves~\eqref{eq:pmccov} and concludes the proof of Lemma~\ref{lem:pmccov2}.
\end{proof}

\begin{figure}[tb]
\begin{center}
\includegraphics{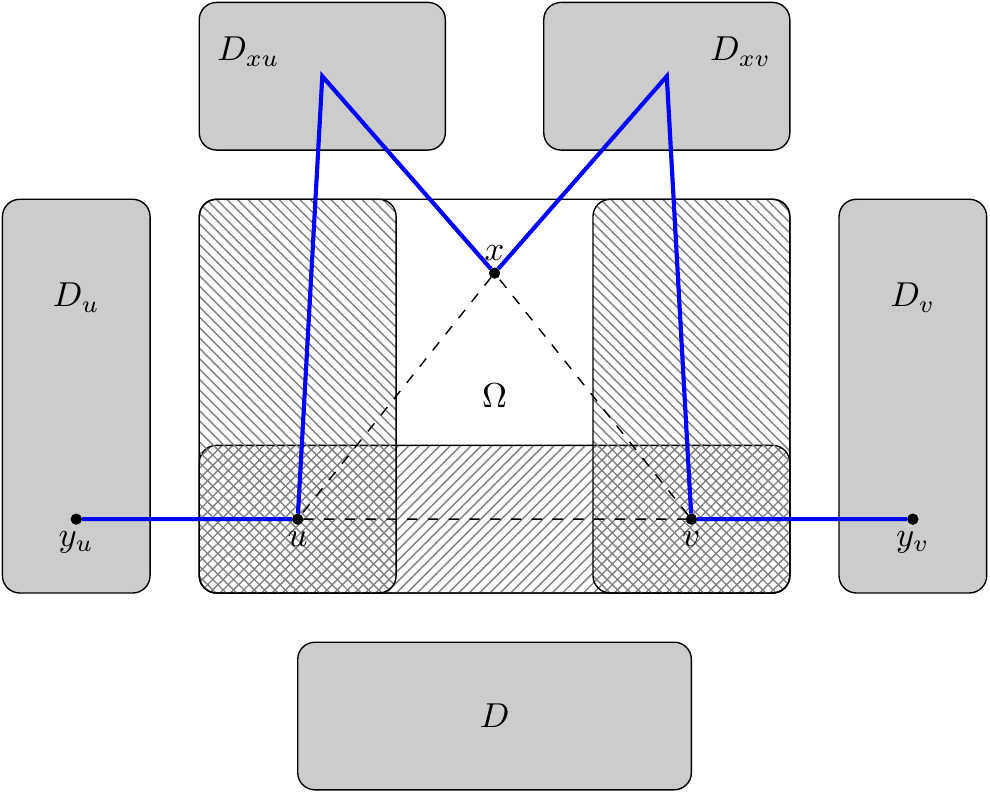}
\caption{A $P_7$ in the proof of Lemma~\ref{lem:pmccov2}. $N(D)$ is marked with north-east strips
  and $N(D_u) \cup N(D_v)$ is marked with north-west strips.
  Note that $N(D_u) \cap N(D_v)$ may be nonempty, which is not depicted on the figure.}\label{fig:pmc}
\end{center}
\end{figure}

\section{Examples}\label{sec:ex}

In this section we discuss two examples showing tightness of the statement of Theorem~\ref{thm:sepcov}: we show that it cannot be generalized to $P_8$-free graphs and that 
a small dominating set of a minimal separator may need to contain elements of the said separator. The examples are modifications of a corresponding example presented
in the conclusions of~\cite{LokshtanovPL18}.

\begin{figure}[tb]
\begin{center}
\includegraphics{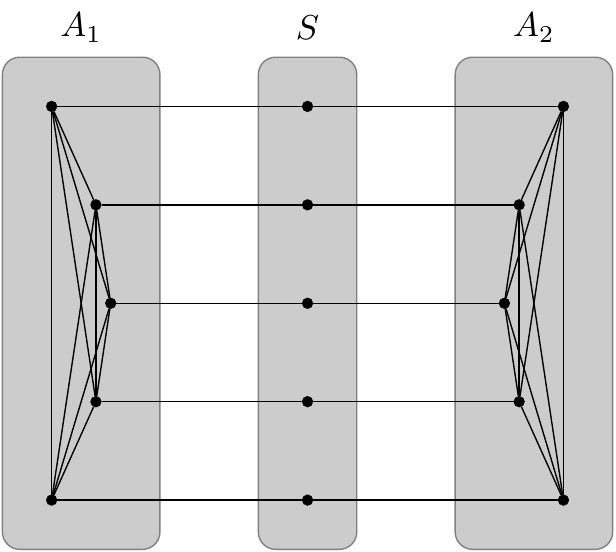}
\caption{The first example for $n=5$. In the second example we turn $S$ into a clique.}\label{fig:ex}
\end{center}
\end{figure}
\paragraph{First example.} Consider the following graph $G$. 
We create three sets of $n$ vertices each, $A_1 = \{a_1^j~|~1 \leq j \leq n\}$,
   $A_2 = \{a_2^j~|~1 \leq j \leq n\}$, and $S = \{s^j~|~1 \leq j \leq n\}$.
We set $V(G) = A_1 \cup A_2 \cup S$. 
For the edge set of $G$, we turn $A_1$ and $A_2$ into cliques
and add edges $s^ja_1^j$ and $s^ja_2^j$, for all $1 \leq j \leq n$. 
This concludes the description of the graph $G$; see Figure~\ref{fig:ex}.
Note that $S$ is a minimal separator in $G$
with $A_1$ and $A_2$ being two full components of $S$.

First, note that for every $v \in V(G)$, $|N[v] \cap S| =1$. Thus, any set dominating
$S$ has to contain at least $n$ vertices. 

Second, note that $G$ is $P_8$-free. To see this, let $P$ be an induced path in $G$.
Since $A_1$ and $A_2$ are cliques, $P$ contains at most two vertices from each $A_i$, $i\in \{1,2\}$,
and these vertices are consecutive on $P$.
Since $S$ is an independent set in $G$, $P$ cannot contain more than one vertex of $S$ in a row.
Hence, $P$ contains at most three vertices of $S$. Consequently $|V(P)| \leq 7$, as desired.
Note that if $n \geq 3$, then there is an induced $P_7$ in $G$, for example
$s^1 - a^1_1 - a^2_1 - s^2 - a^2_2 - a^3_2 - s^3$.

\paragraph{Second example.}
Here, let us modify the graph $G$ from the first example by turning $S$ into a clique.
Still, $S$ is a minimal separator in $G$ with $A_1$ and $A_2$ being two full components of $S$.

First, note that for every $v \in A_1 \cup A_2$, we still have $|N[v] \cap S| =1$.
Thus, any set dominating $S$ that is disjoint with $S$ has to contain at least $n$ vertices.

Second, note that $G$ is $P_7$-free. To see this, observe that $G$ can be partitioned into
three cliques, $A_1$, $A_2$, and $S$, and any induced path in $G$ contains at most two vertices
from each of the cliques. 
Note that if $n \geq 3$, then there is an induced $P_6$ in $G$, for example $a_1^3 - a_1^1 - s^1 - s^2 - a_2^2 - a_2^3$.

\bibliographystyle{abbrv}

\bibliography{references}

\end{document}